\documentclass[12pt]{article}
\usepackage{amsmath,amsfonts,amssymb,latexsym,amsthm}
\usepackage[a4paper,portrait]{geometry}
\usepackage{hyperref}
\usepackage{graphicx,url}

\newtheorem{thm}{Theorem}[section]%
\newtheorem{defi}[thm]{Definition}%

\newtheorem{lemma}[thm]{Lemma}
\newtheorem{rem}[thm]{Remark}%

\title{Confidence regions for the multinomial\\ parameter with small sample
  size}

\author{Djalil~\textsc{Chafa\"\i} and Didier~\textsc{Concordet}}

\date{\small Preprint -- March 2008 -- Revised July 2008 \\ JASA final Accepted Version December 2008}

\begin{document}

\maketitle

\begin{abstract}
  Consider the observation of $n$ iid realizations of an experiment with
  $d\geq2$ possible outcomes, which corresponds to a single observation of a
  multinomial distribution $\mathcal{M}_d(n,\mathbf{p})$ where $\mathbf{p}$ is
  an unknown discrete distribution on $\{1,\ldots,d\}$. In many applications,
  the construction of a confidence region for $\mathbf{p}$ when $n$ is small
  is crucial. This concrete challenging problem has a long history. It is well
  known that the confidence regions built from asymptotic statistics do not
  have good coverage when $n$ is small. On the other hand, most available
  methods providing non-asymptotic regions with controlled coverage are
  limited to the binomial case $d=2$. In the present work, we propose a new
  method valid for any $d\geq2$. This method provides confidence regions with
  controlled coverage and small volume, and consists of the inversion of the
  ``covering collection'' associated with level-sets of the likelihood. The
  behavior when $d/n$ tends to infinity remains an interesting open problem
  beyond the scope of this work.
\end{abstract}


{\footnotesize\tableofcontents}

\section{Introduction}

Consider the observation of $n$ iid realizations $Y_1,\ldots,Y_n$ of an
experiment with $d\geq2$ possible outcomes with common discrete distribution
$p_1\delta_1+\cdots+p_d\delta_d$ on $\{1,\ldots,d\}$, where $\delta_a$ denotes
the Dirac mass at point $a$. This corresponds to a single observation
$\mathbf{X}=(X_1,\ldots,X_d)$ of the multinomial distribution
$$
\mathcal{M}_d(n,\mathbf{p}) %
=\sum_{\stackrel{0\leq k_1,\ldots,k_n\leq n}{k_1+\cdots+k_d=n}}\!\!\!%
\mu_p(k)\delta_{(k_1,\ldots,k_d)}
\quad\text{where}\quad
\mu_{p}(k)=p_1^{k_1}\cdots p_d^{k_d}\frac{n!}{k_1!\cdots k_d!}
$$
where $\mathbf{p}=(p_1,\ldots,p_d)$ and $X_k=\mathrm{Card}\{1\leq i\leq
n\text{ such that }Y_i=k\}$ for every $1\leq k\leq d$. Here $d$ is known,
$\mathbf{X}$ is observed, and $\mathbf{p}$ is unknown. The present article
deals with the problem of constructing a confidence region for $\mathbf{p}$
from the single observation $\mathbf{X}$ of $\mathcal{M}_d(n,\mathbf{p})$, in
the \emph{non-asymptotic} situation where $n$ is \emph{small}. More precisely,
let
$$
\Lambda_d=\{(u_1,\ldots,u_d)\in[0,1]^d\text{ such that } u_1+\cdots+u_d=1\}
$$
be the simplex of probability distributions on $\{1,\ldots,d\}$. The
observation $\mathbf{X}\sim\mathcal{M}_d(n,\mathbf{p})$ lies in the discrete
simplex
\begin{equation}\label{eq:Emultinom}
  E_{d}=\left\{(x_{1},\ldots,x_{d}) \in \{0,\ldots,n\}^{d} %
  \ \text{ such that }\ x_{1}+\cdots+x_{d}=n\right \}.
\end{equation}
From the single observation $\mathbf{X}$ and for some prescribed level
$\alpha\in(0,1)$, we are interested in the construction of a random region
$R_\alpha(\mathbf{X})\subset\Lambda_d$ depending on $\mathbf{X}$ and $\alpha$
such that
\begin{itemize}
\item \emph{the coverage probability has a prescribed lower bound} %
  \begin{equation}\label{eq:cov}
    \mathbb{P}(\mathbf{p}\in R_\alpha(\mathbf{X}))\geq 1-\alpha
  \end{equation}
\item \emph{the volume of $R_\alpha(\mathbf{X})$ in $\mathbb{R}^d$ is as small
    as possible.}
\end{itemize}
These two properties are the most important in practice. We propose to solve
this problem by defining the ``level-set'' confidence region $R_\alpha\left
  (\mathbf{X}\right )\subset\Lambda_d$ given by
\begin{equation}\label{eq:our-region}
R_\alpha\left (\mathbf{X}\right ) %
= \left \{\mathbf{p}\in \Lambda_d\ \ \text{ such that }\ %
  \ \mu_{\mathbf{p}}(\mathbf{X})%
  \geq u(\mathbf{p},\alpha)\right\}
\end{equation}
where 
$$
u(\mathbf{p},\alpha)= \sup\biggr\{ u\in [0,1] \ \text{ such that }\ %
\!\!\!\sum_{\substack{k\in E_d\\\mu_{\mathbf{p}}(k)\geq u}} %
\mu_{\mathbf{p}}(k) %
\geq 1-\alpha \biggr\}.
$$
One can check that this confidence region \eqref{eq:our-region} contains
always the maximum likelihood estimator $n^{-1}\mathbf{X}$ of $\mathbf{p}$.
Moreover, this region can be easily computed numerically, \emph{i.e.} for each
value of $\mathbf{p}$ one may compute $u(\mathbf{p},\alpha)$ and compare it to
$\mu_{\mathbf{p}}(\mathbf{X})$. Furthermore, it fulfills \eqref{eq:cov}, and
the numerical computations presented in Section \ref{se:comp-expl} show that
it has small volume and actual coverage often close to $1-\alpha$ at least for
$d=2$ and $d=3$. In fact, this region is a special case of a generic method of
construction based on \emph{covering collections}. The concept of covering
collections is presented in Section \ref{se:covering-collections} and
encompasses as another special case the classical Clopper-Pearson interval and
its multivariate extensions. On the other hand, it is well known (see for
instance Remark \ref{rm:tests}) that a natural correspondence via inversion
exists between confidence regions with prescribed coverage and families of
tests with prescribed level. However, this correspondence is a simple
translation and does not give any clue to construct regions with small volume.

Two kinds of methods for the construction of a confidence region for
$\mathbf{p}$ can be found in the literature (see for instance
\cite{brown2,brown,blyth3,blyth1,blyth2,newcombe} for reviews). The first
methods give confidence regions with small volume but fail to control the
prescribed coverage (e.g. Bayesian methods with Jeffrey prior, Wald or Wilson
score methods based on the Central Limit Theorem, Bootstrapped
regions,\ldots), and the second control the prescribed coverage but have too
large volume to be useful (e.g. concentration methods based on
Hoeffding-Bernstein inequalities, Clopper-Pearson type methods, \ldots). Note
that the discrete nature of the multinomial distribution produces a staircase
effect which makes it difficult to construct non-asymptotic regions with
coverage equal exactly to $1-\alpha$. For a discussion of such aspects, we
refer for instance to Agresti et al. \cite{MR1628435,MR1814845,MR1977232}. In
general, it seems reasonable to expect a coverage of at least $1-\alpha$,
without being too conservative, while maintaining the volume as small as
possible. Here the term conservative means that the coverage is greater than
$1-\alpha$. Even when $d=2$ and $n$ is large but finite, the confidence
regions built from asymptotic approaches based on the Central Limit Theorem
have a poor and uncontrolled coverage. It is also the case for bootstrapped
versions which only improve the coverage probability asymptotically (see
\cite{MR1323066,MR2077050,MR1736447,MR1145237,garcia-perez,wilson}). For the
binomial case $d=2$, one of the best known method is due to Blyth \& Still
\cite{blyth3} and combines various approaches. To our knowledge, the available
methods for the general multinomial case $d>2$ are unfortunately asymptotic or
Bayesian, which explains their poor performances in terms of coverage or
volume when $n$ is small (see
\cite{MR1893326,MR1325142,MR626900,MR1736447,MR2077050}).

The coverage of our region \eqref{eq:our-region} is strictly controlled since
it fulfills \eqref{eq:cov} whatever the values of $d$ and $n$. However, this
says nothing about the actual coverage and the actual mean volume. The
comparisons presented in Section \ref{se:comp-expl} suggest that our region
for $d=2$ is comparable to the Blyth \& Still region in terms of actual
coverage and actual mean volume. For $d=3$, the Blyth \& Still method is no
longer available, and our region seems to have an actual coverage close to the
prescribed level while maintaining a volume comparable to the asymptotic
region constructed with the score method based on the Central Limit Theorem.
Section \ref{se:comp-expl} provides two concrete examples, one for $d=3$ and
another one for $d=4$ in relation to the $\chi^2$-test. The article ends with
a final discussion.

\section{Covering collections}
\label{se:covering-collections}

The aim of this section is to introduce the notion of \emph{covering
  collection}, which allows confidence regions to be built in a general
abstract space. Let us consider a random variable $X:(\Omega,\mathcal{A})
\rightarrow (E, \mathcal{B}_{E})$ having a distribution $\mu_{\theta^{*}}$
where $\theta^{*}\in\Theta$. For some $\alpha\in(0,1)$, we would like to
construct a confidence region $R_{\alpha}(X)$ for $\theta^{*}$ with a coverage
of at least $(1-\alpha)$, from a single realization of $X$. In other words,
\begin{equation}\label{eq:coverage}
  \mathbb{P}\left (\theta^{*}\in R_\alpha\left (X\right )\right )\geq 1-\alpha.
\end{equation}

\begin{defi}[Covering collection]
  A \emph{covering collection} of $E$ is a collection of measurable events
  $(A_{k})_{k\in \mathcal{K}}\subset\mathcal{B}_{E}$ such that
  \begin{itemize}
  \item $\mathcal{K}$ is totally ordered and has a minimal element and a
    maximal element;
  \item if $k\leq k'$ then $A_k\subset A_{k'}$ with equality if and only if
    $k=k'$;
  \item $A_{\min(\mathcal{K})}=\emptyset$ and $A_{\max(\mathcal{K})}=E$.
  \end{itemize}
\end{defi}

For instance, for $E=\{0,1,\ldots,n\}$, the sequence of sets
$$
\emptyset, \{\sigma(0)\}, \{\sigma(0), \sigma(1)\}, \ldots,
\{\sigma(0),\sigma(1),\ldots,\sigma(n)\}=E
$$
is a covering collection of $E$ for any permutation $\sigma$ of $E$. For
$E=\mathbb{R}$, the collection $(A_t)_{t\in\overline{\mathbb{R}}}$ where
$\overline{\mathbb{R}}=\mathbb{R}\cup\{-\infty,+\infty\}$ defined by
$A_{-\infty}=\emptyset$, $A_t=(-\infty,t]$ for every $t\in\mathbb{R}$, and
$A_{+\infty}=\mathbb{R}$ is a covering collection of $E$. Many other choices
are possible, like $A_t=[-t,+t]$ or $A_t=[t,+\infty)$. We can recognize the
usual shapes of the confidence regions used in univariate Statistics.

\begin{thm}[Confidence region associated with a covering collection]\label{th:conf_region}
  Let $(A_k)_{k\in\mathcal{K}}$ be a covering collection of $E$, and $k_X$ be
  the smallest $k\in \mathcal{K}$ such that $X\in A_k$. For every
  $\alpha\in(0,1)$, the region $R_\alpha(X)$ defined below satisfies to
  \eqref{eq:coverage}.
  \begin{equation}\label{eq:conf_int}
    R_\alpha\left (X\right ) %
    = \left \{\theta\in\Theta\ \text{ such that }\ %
      \mu_\theta(A_{k_X}) \geq  \alpha \right \}.
  \end{equation}
\end{thm}

\begin{proof}
  For every $\theta\in \Theta$, let $k_{\alpha}(\theta)$ be the largest
  $k\in\mathcal{K}$ such that $\mu_{\theta}(A_k)<\alpha$. With this definition
  of $k_{\alpha}(\cdot)$, we then have
  $$
  x\in A_{k_\alpha(\theta)} %
  \quad\text{if and only if}\quad %
  \mu_{\theta}(A_{k_x}) < \alpha.
  $$
  Thus we have
  \begin{align*}
  \mathbb{P}\left (\theta^{*}\in R_\alpha(X) \right ) %
  &= \mathbb{P}\left (\mu_{\theta^{*}}(A_{k_X}) \geq \alpha \right ) \\
  &= \mathbb{P}\left ( X\notin A_{k_\alpha(\theta^{*})}\right) \\
  &=1-\mu_{\theta^{*}}\left ( A_{k_\alpha(\theta^{*})} \right ) \\
  &\geq 1-\alpha.
  \end{align*}
\end{proof}

These confidence regions are highly dependent on the chosen covering
collection $(A_{k})_{k\in\mathcal{K}}$. Each choice of covering collection
gives a particular region $R_{\alpha}(X)$. Note that a small value of $k_X$
gives a small set $A_{k_X}$ and thus leads to a confidence region with a small
volume. For instance, assume that we have two realizations $x_{1}$ and $x_{2}$
of $X$ with $k_{x_{1}}<k_{x_{2}}$. For a given sequence
$(A_{k})_{k\in\mathcal{K}}$, we have $A_{k_{x_{1}}}\subset A_{k_{x_{2}}}$ and
thus $R_\alpha\left (x_{1}\right )\subset R_\alpha\left (x_{2}\right )$. It is
tempting to choose the covering collection $(A_{k})_{k\in\mathcal{K}}$ in such
a way that $k_{X}$ is as small as possible. Unfortunately, with such a choice,
the covering collection $(A_{k})_{k\in\mathcal{K}}$ could be random and the
coverage of the associated region could be less than the prescribed level
$1-\alpha$.

Note that the set $A_{k_X}$ can be empty, which means that a confidence region
cannot be built with the sequence $(A_{k})_{k\in \mathcal{K}}$. In contrast,
the case where $A_{k_X}=E$ leads to the trivial region $R_\alpha(X)=\Theta$.
In the case where $A_{k_X}=\{X\}$, we have
$\mu_\theta(A_{k_X})=\mu_\theta(\{X\})$, which is the likelihood of $X$ at
point $\theta$, and the region $R_\alpha(X)$ corresponds to the complement of
a level-set of the likelihood.

The following symmetrization lemma allows (for instance) the construction of
two-sided confidence intervals from one-sided confidence intervals. We use it
in Section \ref{ss:clopper-pearson} to interpret the Clopper-Pearson
confidence interval as a special case of the covering collection method.

\begin{lemma}[Symmetrization]\label{le:symmetrization}
  Consider a covering collection $(A_k)_{0\leq k\leq \kappa}$ of $E$. For
  every $0\leq k\leq \kappa$ let us define ${A'}_k=E\setminus A_{\kappa-k}$.
  For any $\theta\in\Theta$, any $X\sim\mu_\theta$, and any $\alpha\in(0,1)$,
  we construct
  $$
  R_{\frac{1}{2}\alpha}%
  =\left\{\theta\in\Theta;\ \mu_\theta(A_{k_X})>\frac{1}{2}\alpha\right\} %
  \quad\text{and}\quad %
  R'_{\frac{1}{2}\alpha}%
  =\left\{\theta\in\Theta;\ \mu_\theta(A'_{k'_X})>\frac{1}{2}\alpha\right\}
  $$
  where $k'_X$ is built from $({A'_{k}})_{0\leq k\leq \kappa}$ as $k_X$ from
  $(A_k)_{0\leq k\leq \kappa}$ and $A'_{k'_X}=E\setminus A_{k_X-1}$. Then
  $$
  R_{\frac{1}{2}\alpha}\cap R'_{\frac{1}{2}\alpha} 
  $$
  is a confidence region with coverage greater than or equal to $1-\alpha$.
\end{lemma}

\begin{proof}
  We have
  $\mu_\theta(A_{k_X})+\mu_{\theta}(A'_{{k}_X})=1+\mu_\theta(\{X\})\geq 1$ and
  thus $R_{\frac{1}{2}\alpha}$ and $R'_{\frac{1}{2}\alpha}$ have disjoint
  complements. The conclusion follows now from a general fact: if $R_1$ and
  $R_2$ are two confidence regions with a coverage of at least
  $1-\frac{1}{2}\alpha$ such that $R_1\cup R_2=E$ (equivalently
  $R_1^c=\Theta\setminus R_1$ and $R_2^c=\Theta\setminus R_2$ are disjoint),
  then $R_1^c$ and $R_2^c$ are disjoint and thus $R_1\cap R_2 = (R_1^c\cup
  R_2^c)^c$ is a confidence region with a coverage of at least $1-\alpha$.
\end{proof}

\begin{rem}[Discrete case and staircase effect]
  Let $(A_k)_{k\in\mathcal{K}}$ be a covering collection of a finite set $E$.
  Due to staircase effects, the coverage of the confidence regions constructed
  from this covering collection cannot take arbitrary values in $(0,1)$. These
  staircase effects can be reduced by using a fully granular collection for
  which $\mathrm{Card}(\mathcal{K})=\mathrm{Card}(E)$. The term \emph{fully
    granular} means that the elements of the collection are obtained by adding
  the points of $E$ one by one. It is impossible to remove completely the
  staircase effects when $E$ is discrete, while maintaining a prescribed lower
  bound on the coverage.
\end{rem}

\begin{rem}[Reverse regions]
  For the region ${R}_\alpha(X) = \{ \theta\in\Theta; \mu_\theta(A_{k_X}) \leq 1-\alpha\}$ we
  have
  $$
  \mathbb{P}({R}_\alpha) %
  =\mathbb{P}(\mu_\theta(A_{k_X})\leq 1-\alpha) %
  =\mathbb{P}(X\in A_{k_{1-\alpha}}) %
  =\mu_\theta(A_{k_{1-\alpha}}) \leq 1-\alpha.
  $$
\end{rem}

\begin{rem}[Link with tests]\label{rm:tests}
Let us recall briefly the correspondence between confidence regions and
statistical tests (we refer to \cite[Section 48]{MR1712750} for further
details). Consider a parametric model $(\mu_\theta)_{\theta\in\Theta}$ with
data space $\mathcal{X}$. For any fixed $\theta_0\in\Theta$, the test problem
of $H_0:\theta=\theta_0$ versus $H_1:\theta\neq\theta_0$ with level
$\alpha\in(0,1)$ corresponds to the construction of an acceptance region
$C_\alpha(\theta_0)\subset\mathcal{X}$ such that
$$
\mu_{\theta_0}(C_\alpha(\theta_0))\geq 1-\alpha.
$$
The construction of a confidence region for $\theta_0$ can be done by
inversion (\emph{i.e.} by collecting the values of $\theta_0$ for which $H_0$ is
accepted). Namely, for every $x\in\mathcal{X}$, one can define the region
$R_\alpha(x) \subset\Theta$ by
$$
R_\alpha(x) %
=\{\theta\in\Theta\text{ such that } x\in C_\alpha(\theta)\} .
$$
Now if $X\sim\mu_{\theta_0}$ then
$$
\mathbb{P}(\theta_0\in R_\alpha(X))
= \mathbb{P}(X\in C_\alpha(\theta_0)) %
=\mu_{\theta_0}(C_\alpha(\theta_0)) %
\geq 1-\alpha. 
$$
This shows that for any fixed $\theta_0\in\Theta$, the set
$R_\alpha(X)\subset\Theta$ is a confidence region for $\theta_0$ when
$X\sim\mu_{\theta_0}$. Conversely, if for every $\theta_0\in\Theta$ and every
$x\in\mathcal{X}$ one has a region $R_\alpha(x)\subset\Theta$ such that
$\mathbb{P}(\theta_0\in R_\alpha(X))\geq 1-\alpha$ when $X\sim\mu_{\theta_0}$,
then one can construct immediately a test for $H_0:\theta=\theta_0$ versus
$H_1:\theta\neq\theta_0$ with acceptance region
$$
C_\alpha(\theta_0) %
=\{x\in\mathcal{X}\text{ such that }\theta_0\in R_\alpha(x)\}.
$$
Note that this correspondence between confidence regions and statistical tests
can be extended to the composite case $H_0:\theta\in\Theta_0$ versus
$H_1:\theta\not\in\Theta_0$ where $\Theta_0\subset\Theta$.
\end{rem}


\subsection{The level-sets regions}
\label{ss:level-sets}

In this section, we show that the ``level-sets'' confidence region
\eqref{eq:our-region} is a special case of the covering collection method. It
is easier to consider here a decreasing covering collection (the corresponding
version of Theorem \ref{th:conf_region} is immediate). Let us consider a
random variable $X:(\Omega,\mathcal{A}) \rightarrow (E, \mathcal{B}_{E})$ with
law $\mu_{\theta^{*}}$ where $\theta^{*}\in\Theta$. For every $u\geq0$ and
$\theta\in \Theta$, let us define
$$
A(\theta,u)= \{ x\in E \text{ such that }\mu_\theta(x)\geq u\}.
$$
For every $\theta\in \Theta$, the collection $(A(\theta,u))_{u\geq 0}$ is
decreasing with $A(\theta,0) = E$ and there exists $u_{\rm max}$ that can be
equal to $+\infty$ such that $A(\theta,u_\text{max}) = \emptyset$. Also,
$(A(\theta,u_\text{max}-u))_{u\in[0,u_\text{max}]}$ is a covering collection
of $E$. Next, define
$$
u(\theta,\alpha)= \sup\left \{ u\in [0,u_{\rm max}] \text{ such that } %
  \mu_\theta(A(\theta,u) ) \geq 1-\alpha\right \}
$$
and
$$
K(\theta,\alpha) = A(\theta, u(\theta,\alpha)).
$$
We would like to construct a confidence region for $\theta^*$ from the
observation of $X\sim \mu_{\theta^*}$. If
\begin{equation}
  R_\alpha\left (X\right ) %
  = \left \{\theta\in \Theta \text{ such that } X\in K(\theta,\alpha)\right \}
\end{equation}
then
$$
\mathbb{P}\left (\theta^{*}\in R_\alpha\left (X\right )\right )%
=\mathbb{P}\left (X\in K(\theta^{*},\alpha)\right
)=\mu_{\theta^{*}}(K(\theta^{*},\alpha))\geq 1-\alpha.
$$
This shows that $R_\alpha(X)$ is a confidence region for $\theta^*$ with a
coverage of at least $1-\alpha$. Let clarify the expression of the confidence
region for the general multinomial case where
$\mathbf{X}\sim\mathcal{M}_d(n,\mathbf{p})$ with $\mathbf{p}\in\Lambda_d$ and
$d\geq2$. Here the value of $\mathbf{p}$ used for the observed data
$\mathbf{X}$ plays the role of $\theta^*$. We have $\Theta=\Lambda_d$,
$E=E_{d}$ as described by \eqref{eq:Emultinom},
$\mu_{\boldsymbol\theta}=\mathcal{M}_d(n,\boldsymbol{\theta})$, and $u_{\rm
  max}=1$. For every $\alpha\in(0,1)$, the confidence region given by the
level-sets method is expressed as in (\ref{eq:our-region}) given in the
introduction.

\subsubsection*{Optimality}

Let us focus on the case where $E$ is a finite set. The confidence region
constructed above is not optimal among the $1-\alpha$ conservative regions and
thus could be improved by a more detailed analysis. Let us first note that by
its very construction, for each $\theta\in \Theta$, $K(\theta,\alpha)$ is
minimal with respect to its cardinality that is, a set $B(\theta,\alpha)$ does
not exist so that $\mu_{\theta}(B(\theta,\alpha))\geq 1-\alpha$ and ${\rm
  card}(B(\theta,\alpha))<{\rm card} (K(\theta,\alpha))$. However, in some
circumstances, sets $L(\theta,\alpha)$ may exist with the same cardinality as
$K(\theta,\alpha)$ so that $\mu_{\theta}(K(\theta,\alpha))\geq
\mu_{\theta}(L(\theta,\alpha))\geq 1-\alpha$. The following theorem gives a
condition that allows conservative sets to be built but with a coverage closer
to $1-\alpha$ than the coverage of $R_\alpha\left (X\right )$. For all
$\alpha\in[0,1]$ and $\theta\in \Theta$, let us denote $\gamma(\theta,\alpha)
= 1-\mu_{\theta}\left (K(\theta,\alpha)\right )$ and let us note that
$\gamma(\theta,\alpha)\leq\alpha$.

\begin{thm}
  If for each $\theta\in \Theta$ there exists two subsets $V\left
    (\theta,\alpha\right )\subset K(\theta,\alpha)$ and $W\left
    (\theta,\alpha\right )\subset E\backslash K(\theta,\alpha)$ with the
  same cardinality so that
  $$
  \alpha-\gamma(\theta,\alpha) %
  \geq\mu_{\theta}\left (V (\theta,\alpha )\right )%
  -\mu_{\theta}\left (W (\theta,\alpha )\right )>0,
  $$
  then there exists a set $T_\alpha\left (X\right )\neq R_\alpha\left (X\right
  )$ so that
  $$
  1-\alpha%
  \leq\mathbb{P}\left(\theta^{*}\in T_\alpha\left (X\right )\right)%
  <\mathbb{P}\left (\theta^{*}\in R_\alpha\left (X\right )\right ).
  $$
\end{thm}

\begin{proof}
  Let us consider the set $L\left (\theta,\alpha\right )=K\left
    (\theta,\alpha\right )\backslash V\left (\theta,\alpha\right )\bigcup
  W\left (\theta,\alpha\right )$ and note that thanks to the conditions
  imposed for the sets $V$ and $W$ we have for all $\theta\in \Theta$,
  $$
  1-\alpha\leq\mu_{\theta}\left (L(\theta,\alpha)\right )<\mu_{\theta}\left
    (K(\theta,\alpha)\right ).
  $$
  Now, with $T_\alpha\left(X\right )=\left\{\theta\in \Theta;X\in
    L(\theta,\alpha)\right \}$ we have
  \begin{equation*}
    \begin{array}{ll}
      \mathbb{P}\left (\theta^{*}\in T_\alpha\left (X\right )\right )&= \mathbb{P}\left (X\in L(\theta^{*},\alpha)\right )\\
      {}&=\mathbb{P}\left (X\in K\left (\theta^{*},\alpha\right )\setminus V\left (\theta^{*},\alpha\right )\bigcup W\left (\theta^{*},\alpha\right )\right )\\
      {}&=1-\gamma\left (\theta^{*},\alpha\right )-\mu_{\theta^{*}}\left ( V\left (\theta^{*},\alpha\right )\right )+\mu_{\theta^{*}}\left ( W\left (\theta^{*},\alpha\right )\right )\\
      {}&\leq 1-\gamma\left (\theta^{*},\alpha\right ).\\
    \end{array}
  \end{equation*}
  On the other hand, we have already seen that for all $\theta\in \Theta$,
  $$
  1-\alpha\leq\mu_{\theta}\left (L(\theta,\alpha)\right ).
  $$
  This last inequality holds true when $\theta=\theta^{*}$ and thus
  $$
  1-\alpha\leq\mu_{\theta^{*}}\left (L(\theta^{*},\alpha)\right )=\mathbb{P}\left
    (\theta^{*}\in T_\alpha\left (X\right )\right ).
  $$
\end{proof}

This theorem can be used to build less conservative confidence sets than
$R_{\alpha}(X)$. A convenient way to proceed is to take
$V(\theta,\alpha)=\{y\}$ where $y$ is such that
$$
\mu_{\theta}(y)=\min_{z\in
  K\left (\theta,\alpha\right )}\mu_{\theta}(z)
$$
and to iteratively try several sets $W^{k}$ as follows. Set
$W^{0}(\theta,\alpha)=\emptyset$, and at iteration $k\geq 1$, set
$W^{k}(\theta,\alpha)=\{w_{k}\}$ and $L^{k}(\theta,\alpha)=K\left
  (\theta,\alpha\right )\setminus V\left (\theta,\alpha\right )\bigcup
W^{k}\left (\theta,\alpha\right )$ where
$$w_{k}=\arg\max_{z\in L^{k-1}(\theta,\alpha) }\mu_{\theta}(z).$$
This process is iterated until the set $L^{k}(\theta,\alpha)$ is such that
$\mu_{\theta}\left (L^{k}(\theta,\alpha)\right )-(1-\alpha)$ is non-negative
and minimum.

Since for $\theta\in \Theta$ there may exist $x\neq y$ with
$\mu_{\theta}(x)=\mu_{\theta}(y)$, there also may exist several sets
$(L^{i}(\theta,\alpha))_{i}$ which have the same mass
$\mu_{\theta}(L^{i}(\theta,\alpha))=1-\delta(\theta,\alpha).$ Several
confidence sets with the same coverage can thus be derived using these sets. A
simple way to choose between these concurrent confidence sets is to adopt the
one that optimizes a criterion such as having a minimum volume (for the
Lebesgue measure).


\subsection{The Clopper-Pearson regions}
\label{ss:clopper-pearson}

Consider the binomial case $d=2$ for which $\mathbf{p}=(p_1,1-p_1)$. The well
known Clopper-Pearson interval for $p_1$ relies on the exact distribution of
$X_1$ in the binomial case \cite{clopper-pearson, julious, chen}. It was
considered for a long time as outstanding. This interval $[L, U]$ is given by
\begin{equation}\label{eq:clopper_limits}
  \begin{cases}
    L&=\inf \left \{ \theta\in[0,1] \text{ such that }
      \sum_{i=x_{1}}^n\binom{n}{i}\theta^{i}(1-\theta)^{n-i} \geq \frac{1}{2}\alpha\right
    \} \\
    U&=\sup \left \{ \theta\in[0,1] \text{ such that }
      \sum_{i=0}^{x_{1}}\binom{n}{i}\theta^{i}(1-\theta)^{n-i} \geq
      \frac{1}{2}\alpha\right \}.
  \end{cases}
\end{equation}

It has been shown that the Clopper-Pearson interval is often conservative.
Also, some continuity corrections have been proposed, and give the so called
``mid-p interval'', see \cite{berry} for a review. This trick reduces the
staircase effect but the coverage probability can be less than $1-\alpha$. The
Beta-Binomial correspondence (see Lemma \ref{le:Beta-Binom} below) shows that
the left and right limits $L$ and $R$ of the Clopper-Pearson confidence
interval \eqref{eq:clopper_limits} are the $\frac{1}{2}\alpha$ and
$(1-\frac{1}{2}\alpha)$ quantiles of the Beta distribution
$\mathrm{Beta}\left(X_1;n-X_1+1\right)$.

\begin{lemma}[Beta-Binomial correspondence]\label{le:Beta-Binom}
  If $X\sim\mathrm{Binom}(n,p_1)$ with $p_1\in[0,1]$ and $0\leq k\leq n$ and
  $B\sim\mathrm{Beta}(k,n-k+1)$ then the following identity holds true.
  \begin{equation}\label{eq:binbeta}
    \mathbb{P}(X\geq k)=\mathbb{P}(B\leq p_1).
  \end{equation}
\end{lemma}

\begin{proof}
  We briefly recall here the classical proof (see \cite[page 68]{MR1712750}).
  Let $U_1,\ldots,U_n$ be iid uniform random variables on $[0,1]$ and
  $U_{(1)}\leq\cdots\leq U_{(n)}$ be the reordered sequence. If we define
  $V_{p_1}=\sum_{i=1}^n\mathrm{I}_{\{U_i\leq p_1\}}$ then
  $V_{p_1}\sim\mathrm{Binom}(n,p_1)$ and $U_{(k)}\sim\mathrm{Beta}(k,n-k+1)$
  and for every $1\leq k\leq n$, $V_{p_1} \geq k$ if and only if $U_{(k)}\leq
  p_1$.
\end{proof}

The confidence interval obtained by the level-sets method does not coincide
with the classical Clopper-Pearson confidence interval. Let us show why the
Clopper-Pearson confidence interval can be considered as a special case of the
method based on covering collections. Recall that we are in the case where
$d=2$ and $X_1\sim\mathrm{Binom}(n,p_1)$ for some unknown $p_1\in[0,1]$. This
can also be written $(X_1,n-X_1)\sim\mathcal{M}_2(n,(p_1,1-p_1))$. The
unidimensional nature of $E=\{0,\ldots,n\}$ suggests the following two
covering collections $(A_k^1)_{k\in E}$ and $(A_k^2)_{k\in E}$ defined by
$A_{0}^{1} =\emptyset$ and $A_{0}^{2} =\emptyset$, and for every $0\leq k\leq
n$,
$$
A^1_{k+1}=\{0,\ldots,k\}%
\quad\text{and}\quad
A^2_{k+1}=\{n-k,\ldots,n\}.
$$
Here $\mathcal{K}=E$ for both the top-to-bottom and bottom-to-top sequences.
The bottom-to-top sequence $(A^{1}_{k})_{k\in E}$ leads to a $(1-\alpha)$
one-sided confidence interval for $p_1$ given by
\begin{equation}\label{eq:CP-U}
  R^{1}_\alpha\left (X_1\right ) %
  = \left \{\theta\in [0,1] \text{ such that} %
    \sum_{i=0}^{X_1}\binom{n}{i}\theta^{i}(1-\theta)^{n-i} \geq \alpha\right\} %
  = [0,U_{\alpha}(X_1)]
\end{equation}
where
$$
U_{\alpha}(x) %
=\sup \left \{ \theta\in[0,1] \text{ such that} %
  \sum_{i=0}^{x}\binom{n}{i}\theta^{i}(1-\theta)^{n-i} \geq \alpha\right \}.
$$
On the other hand, the top-to-bottom covering collection $(A^{2}_k)_{k\in E}$
leads to a $(1-\alpha)$ confidence interval of $p_1$ given by
\begin{equation}\label{eq:CP-L}
  R^{2}_\alpha\left (X_1\right ) %
  = \left \{\theta\in [0,1] \text{ such that} %
    \sum_{i=X_1}^n\binom{n}{i}\theta^{i}(1-\theta)^{n-i}\geq \alpha\right\} %
  = [L_{\alpha}(X_1);1]
\end{equation}
where
$$
L_{\alpha}(x) %
=\sup \left \{ \theta\in[0,1] \text{ such that} %
  \sum_{i=x}^{n}\binom{n}{i}\theta^{i}(1-\theta)^{n-i} \geq \alpha\right \}.
$$
By virtue of Lemma \ref{le:symmetrization}, we can combine the one-sided
confidence intervals \eqref{eq:CP-U} and \eqref{eq:CP-L} in order to obtain a
two-sided $(1-\alpha)$ confidence interval of $p_1$, which is the two-sided
interval
$$
R^{1}_{\frac{1}{2}\alpha}\left (X_1\right ) %
\bigcap R^{2}_{\frac{1}{2}\alpha}\left (X_1\right ) %
=[L_{\frac{1}{2}\alpha}(X_1);U_{\frac{1}{2}\alpha}(X_1)].
$$
We recognize the Clopper-Pearson interval \eqref{eq:clopper_limits}. The
discrete nature of $E$ precludes the construction of a confidence interval of
$p_{1}$ with coverage exactly equal to $1-\alpha$. Actually, the
Clopper-Pearson interval is not exactly symmetric and there is no guaranty
that
$$
\mathbb{P}\left(p< L_{\frac{1}{2}\alpha}(X_1)\right) %
= %
\mathbb{P}\left(p> U_{\frac{1}{2}\alpha}(X_1)\right).
$$
Our construction via a covering collection immediately provides an extension
of the Clopper-Pearson interval in the general multinomial case where
$\mathbf{X}\sim\mathcal{M}_d(n,\mathbf{p})$ with $\mathbf{p}\in\Lambda_d$ and
$d>2$. This construction consists of labeling the elements of $E_d$ (note that
$\mathrm{Card}(E_{d})=\binom{n+d-1}{d-1}$) and constructing the covering
collection $(A_{k})_{k\in\mathcal{K}}$ which grows by adding the points one
after the other. The choice of the total order on $E_d$ is arbitrary when
$d>2$. Some additional constraints can help to reduce this choice. As
advocated by Casella \cite{casella} for the binomial distribution, the
proposed confidence region $R_{\alpha}\left (X\right )$ should be
\emph{equivariant}, that is not sensitive to the order chosen to label the $d$
categories of the multinomial distribution.

\begin{defi}[Equivariance]
  A confidence region $R_\alpha(X)$ is \emph{equivariant} when
  \begin{equation}\label{eq:equivariance}
    \mathbb{P}\left(\sigma(\theta^{*})\in R_\alpha\left(\sigma(X)\right)\right)%
    =\mathbb{P}\left(\theta^{*}\in R_\alpha\left(X\right )\right)
  \end{equation}
  for every permutation $\sigma$ of $\{1,\ldots,d\}$. In other words, if and
  only if
  $$
  \sigma\left (R_\alpha\left (X\right )\right ) %
  =R_\alpha\left (\sigma(X)\right ).
  $$
\end{defi}

The following lemma gives a criterion of equivariance for covering
collections.

\begin{thm}[Equivariance criterion for covering collections]
  The confidence region $R_\alpha(\mathbf{X})$ constructed from a covering
  collection $(A_k)_{k\in\mathcal{K}}$ is equivariant if and only if $A_k$ is
  invariant by permutation of coordinates for every $k\in\mathcal{K}$.
\end{thm}

\begin{proof}
  Let $\sigma$ be a permutation of $\{1,\ldots,d\}$,
  $\mathbf{i}=(i_{1},\ldots,i_{d})\in E$, and for every $\theta\in\Theta$,
  $$
  \sigma(\theta)=\left(\theta_{\sigma(1)},\ldots, \theta_{\sigma(d)}\right) %
  \quad\text{and}\quad %
  \sigma(\mathbf{i})=\left (i_{\sigma(1)},\ldots, i_{\sigma(d)}\right ).
  $$
  By invariance of $A_k$ by permutation, we have $\mathbf{X}\in
  A_{k}\Leftrightarrow \mathbf{X}\in \sigma(A_{k} )$ and thus
  $k_{\mathbf{X}}=k_{\sigma(\mathbf{X})}$. If $\theta\in\sigma\left
    (R_\alpha\left (\mathbf{X}\right )\right )$ then
  $\mu_{\sigma^{-1}(\theta)}(A_{k_{\mathbf{X}}})\geq\alpha$. But, for every
  $\mathbf{i}\in E$,
  $$
  \mu_{\sigma^{-1}\theta)}(\{\mathbf{i}\}) =\mu_{\theta}(\{\sigma(\mathbf{i})\}).
  $$
  If $A_{k}$ is invariant by permutations, then for every $\mathbf{i}\in A_k$,
  we have $\sigma(\mathbf{i})\in A_{k}$ and consequently
  $$
  \mu_{\sigma^{-1}(\theta)}(A_k) %
  =\mu_{\theta}(\sigma(A_k)) %
  =\mu_{\theta}(A_k).
  $$
  Thus, $\theta\in\sigma\left (R_\alpha\left (\mathbf{X}\right )\right )$ if
  and only if $\mu_{\theta}(A_{k_{\mathbf{X}}})=\mu_{\theta}(A_{k_{\sigma
      (\mathbf{X})}})\geq\alpha$, that is $\theta\in R_\alpha\left
    (\sigma(\mathbf{X})\right )$.
\end{proof}

Equivariance imposes a strong constraint on the covering collection. A large
set $A_{k_{\mathbf{X}}}$ gives a large confidence region. Since confidence
regions with small volume are desirable, it is interesting, when $E$ is
discrete, to consider a covering collection $(A_k)_{k\in\mathcal{K}}$ which
grows by adding the points of $E$ one after the other. Unfortunately, this
method of construction is not compatible with equivariance: the $A_k$ cannot
be invariant by permutations of coordinates. A weaker condition consists of
the existence of a subsequence $(A_{k_{l}})_{l}$ that is invariant by
permutation of coordinates. An example of such a sequence for $d=3$ is given
in Figure \ref{fig:crob0}.

Recall that when $d=2$, the Beta-Binomial correspondence stated in Lemma
\ref{le:Beta-Binom} provides a clear link between the quantiles of the Beta
distribution and the Clopper-Pearson confidence interval. In fact, this can be
seen as a special case of the Dirichlet-Multinomial correspondence valid for
any $d\geq3$ as stated in the following lemma. This makes a link between
Clopper-Pearson regions and Bayesian regions constructed with a Jeffrey prior
(see for instance \cite{MR1835885}). However, the notion of coverage that we
use in the present article is purely frequentist and does not fit with the
Bayesian paradigm without serious distortions.
 
\begin{lemma}[Dirichlet-Multinomial correspondence]
  Let $\mathbf{p}\in\Lambda_d$ and $k_0,k_1,\ldots,k_d$ be such that
  $k_0=0\leq k_1\leq\cdots\leq k_{d-1}\leq n\leq k_{d}=n+1$. If
  $$
  \mathbf{X}\sim\mathcal{M}_d(n,\mathbf{p}) %
  \quad\text{and}\quad %
  \mathbf{D}\sim\mathrm{Dirichlet}_d(k_1-k_0,k_2-k_1,\ldots,k_d-k_{d-1})
  $$
  then the following identity holds true:
  \begin{multline}\label{eq:multinomdiric}
    \mathbb{P}(X_1\geq k_1,X_1+X_2\geq k_2,\ldots,X_1+\cdots+X_{d-1}%
    \geq  k_{d-1})\\
    =\mathbb{P}(D_1\leq p_1,D_1+D_2\leq p_2,\ldots,D_1+\cdots+D_{d-1}%
    \leq p_{d-1}).
  \end{multline}
\end{lemma}

\begin{proof}
  The proof is a direct extension of the Beta-Binomial case given by Lemma
  \ref{le:Beta-Binom}. Let $I_1,\ldots,I_d$ be the sequence of adjacent
  sub-intervals of $[0,1]$ of respective lengths $p_1,\ldots,p_d$,
  $U_1,\ldots,U_n$ be iid uniform random variables on $[0,1]$ and
  $U_{(1)}\leq\cdots\leq U_{(n)}$ be the reordered sequence. For any $1\leq
  r\leq d$, let us define
  $$
  V_{p,r}= \sum_{i=1}^n \mathrm{I}_{\{U_i\in \mathrm{I}_r\}} %
  =\mathrm{Card}\{1\leq i\leq n\ \text{ such that }\ U_i\in\mathrm{I}_r\}.
  $$
  We have $\mathbf{V_p}=(V_{p,1},\ldots,V_{p,r})\sim\mathcal{M}_d(n,\mathbf{p})$. Now,
  for every $0\leq k_1 \leq\cdots \leq k_{d-1}\leq n$,
  $$
  V_{p,1}\geq k_1,\ldots,V_{p,1}+\cdots+V_{p,d-1}\geq k_{d-1} \quad\text{iff}
  \quad U_{(k_1)}\leq p_1,\ldots,U_{(k_{d-1})}\leq p_1+\cdots+p_{d-1}.
  $$
  But by using the notation $U_{(0)}=0$ and $U_{(n+1)}=1$, we have
  $$
  (U_{(1)}-U_{(0)},\ldots,U_{(n+1)}-U_{(n)})
  \sim\mathrm{Dirichlet}_{n+1}(1,\ldots,1).
  $$
  and therefore, by the stability of Dirichlet laws by sum of blocks, with
  $k_0=0$ and $k_d=n+1$,
  $$
  (U_{(k_1)}-U_{(k_0)},\ldots,U_{(k_{d})}-U_{(k_{d-1})})
  \sim\mathrm{Dirichlet}_d(k_1,k_2-k_1,\ldots,k_d-k_{d-1}).
  $$
\end{proof}

\section{Comparisons and examples}
\label{se:comp-expl}

Recall that for every fixed $d\geq2$, $n\geq0$, and $\mathbf{p}\in\Lambda_d$,
a confidence region obtained from $\mathbf{X}\sim\mathcal{M}(n,\mathbf{p})$
provides a single coverage probability and a distribution of volumes. In this
section, we use coverage probabilities and mean volumes to compare the
performance of our level-set method with other methods, in the case where
$d\in\{2,3\}$ and $n\in\{5,10,20,30\}$. We also give two concrete examples,
one for $d=3$ and another one for $d=4$ in relation to the $\chi^2$-test. It
turns out that the regions obtained by the Clopper-Pearson method and its
multinomial extension have non-competitive volumes so we decided to ignore
them in the comparisons.

\subsection{Performances in the binomial case  ($d=2$)}
\label{ss:binomial}

In the binomial case $d=2$, a confidence region for $\mathbf{p}=(p_1,1-p_1)$
is actually a confidence interval for $p_1$. It is well known that the Wald
interval constructed from the Central Limit Theorem has poor coverage even
when $n$ is large but finite \cite{brown2}. It is also widely accepted that
the Wilson score interval \cite{wilson,brown2} or the Blyth-Still interval
\cite{blyth1} should be preferred to the Wald interval. We therefore compared
the performances of the $95\%$-intervals provided by the level-sets method,
the score method, and the Blyth-Still method. We computed the coverages and
the mean widths of the intervals obtained with each method for
$n\in\{5,10,20,30\}$ and for all $p_1\in[0;0.5].$ The results are represented
in figures \ref{fig:crob1} and \ref{fig:crob2} respectively. We can see that
for some values of $p_1$, the coverage of the score method is smaller than the
prescribed level of $0.95$, whereas the coverage of the Blyth-Still interval
and the level-set interval are always greater than or equal to this prescribed
level $0.95$. The coverages obtained with the level-set method are always
closer to the prescribed level except for $n=20,\ p_1\in[0.45,0.48]$ and
$n=30,\ p_1\in[0.38,0.42]$. The differences between the coverages of these
three methods decrease with $n$.

Figures \ref{fig:crob1} and \ref{fig:crob2} show that the score method
provides intervals with excellent mean width but fails to control the
coverage. The level-set method gives intervals that have a slightly narrower
mean width than the one obtained with the Blyth-Still method. This suggests
that the level-set method provides an excellent alternative to the Blyth-Still
method. Moreover, and in contrast to the Blyth-Still method, the level-set
method can still be used when $d>2$.

\subsection{Performances in the trinomial case ($d=3$)}
\label{ss:trinomial}

To our knowledge, the Blyth-Still method has no counterpart for $d>2$. In
addition, the regions obtained by the extended Clopper-Pearson method have
non-competitive volumes. We therefore decided to compare the level-set method
with the natural multidimensional extension of the Wilson score method. We
computed for $d=3$ the coverage probabilities and the mean volumes of the
$95\%$-regions obtained with both methods, for $n\in\{5,10,20\}$. Note that
for the score method, only the trace over $\Lambda_3$ of the regions is used
to compute the volume. The graphics in Figure \ref{fig:crob3} show the
coverage of both methods as well as the difference between their mean volumes.
Whatever the sample size, the coverage of the level-set regions is very close
to $1-\alpha=0.95$. In contrast, the coverages of the score regions can be
much lower than $0.95$. Surprisingly and in contrast with the binomial case
($d=2$), the level-set method here provides confidence regions with mean
volumes that (for $n=5$) are comparable to or smaller than their score's
counterparts! We believe that this because we measure the performance by the
mean volume. The level-set method appears thus to be a reasonable way to build
small confidence sets.

\subsection{Concrete example of the trinomial case ($d=3$)}
\label{ss:ex1}

The present example concerns antibiotics efficacy. A traditional way to
evaluate whether or not an antibiotic can be used for a specific pathogen is
to perform a ``susceptibility testing''. In such an experiment, different
isolates of a given pathogen are classified as ``Sensible'', ``Intermediate''
or ``Resistant'' according to the antibiotics ability to stop their growth.
Here, ten different isolates of Escherichia coli were tested with ampicillin.
The following results were obtained : 8 isolates were Sensible, 2 Intermediate
and 0 were Resistant. The count $\mathbf{x}=(8,2,0)$ can be seen as the
realization of $\mathbf{X}\sim\mathcal{M}(10,\mathbf{p})$ where
$\mathbf{p}=(p_{1},p_{2},p_{3})$ denotes the probability of a given isolate
belonging to each of the different classes. We calculated a $95\%$-confidence
region of $\mathbf{p}$ using the level-set method (Figure \ref{fi:ex1}). This
region suggests that even if none of the $10$ tested isolates was observed to
be resistant, up to $30\%$ of resistant and $20\%$ of intermediate isolates
will be still possible. This confidence region does not contain the situation
where all the isolates are sensible and it is thus unlikely that this
antibiotic works all the time when it meets this pathogen.

\subsection{Concrete example of the quadrinomial case ($d=4$)}
\label{ss:ex2}

The present example is simply a $\chi^{2}$-test for independence. It deals
with the difference in behavior of male and female veterinary students with
respect to smoking habits. The following result was observed in a group of 12
veterinary students in Toulouse:

\begin{center}
\begin{tabular}{lcc}
  \hline
  {}&Smokers& Non-smokers\\  \hline
  Female & 3 & 8 \\
  Male & 10 & 5 \\
  \hline
\end{tabular}
\end{center}

The $\chi^{2}$-test rejects independence with a $P$-value $0.047$ and suggests
that more males than females smoke. This $P$-value is close to the critical
threshold of $0.05$ and was obtained with a small sample size. Therefore, one
can question whether this result can be trusted. A possible solution is to
build a confidence region. The table above can be seen as the realization
$\mathbf{x} = (3, 8, 10, 5)$ of a multinomial random variable $\mathbf{X}\sim
\mathcal{M}(26, \mathbf{p})$ with $\mathbf{p}=(p_{1},p_{2},p_{3},p_{4})$. If
the smoking habit and the gender are independent then $\mathbf{p}$ belongs to
$$
H_{0} %
=\left \{\mathbf{q}\in\Lambda_{4} \text{ such that } \mathbf{q} %
=\left (uv, (1-u)v,u(1-v),(1-u)(1-v)\right ) %
\text{ and } %
 (u,v)\in[0,1]^{2}\right \}.
$$
Since $p_{4}=1-p_{1}-p_{2}-p_{3}$, one can draw a graphic with only $p_{1},
p_{2},p_{3}$. Figure \ref{fig:chi2} shows (in green) the $95\%$ confidence
region for $\mathbf{p}$ built with the level-set method. The surface
corresponds to the null hypothesis $H_{0}$. The red area is the acceptance
region of the $\chi^{2}$-test. It turns out that $\mathbf{\hat p}=(3/26, 8/26,
10/26)$ does not belong to the acceptance region of the $\chi^{2}$-test.
However, the $95\%$-region for $\mathbf{p}$ obtained with the level-set method
cuts $H_{0}$. Therefore, according to Remark \ref{rm:tests} and in contrast to
the result given by the $\chi^{2}$-test, the independence hypothesis is not
rejected.

The $95\%$ level-set confidence region provides the following $95\%$
confidence interval for the \emph{odd-ratio}: $[0.024; 1.712]$. On the other
hand, the inversion of Fisher's exact test gives the $93.7\%$ interval
$[0.187; 2.625]$. This suggests that the level-set approach is less
conservative, probably due to the fact that Fisher's exact conditions on row
and column totals increases the discreteness of the problem.

\section{Final discussion}
\label{se:discussion}

The general concept of ``covering collection'' allows the construction of
confidence regions with controlled coverage, including the classical
Clopper-Pearson interval for the binomial and its multinomial extensions. The
covering collection construction involves an arbitrary growing collection of
sets in the data space. Our ``level-set'' confidence regions are obtained by
using a special collection based on level-sets of the data distribution. The
level-set regions for the multinomial parameter can be easily computed for any
$d$ and $n$. It turns out that they have excellent coverage probabilities and
mean volumes for $d\in\{2,3\}$ and $n\leq 30$. They are in particular
competitive with the famous Blyth-Still intervals for $d=2$. Also, we
recommend the level-set method, even if it can be computationally expensive
when $d$ is large. The behavior of these confidence regions when the ratio
$d/n$ tends to infinity is a very interesting open problem. In this extreme
case, the observation $\mathbf{X}$ is sparse and belongs to the boundary of
the observation simplex $E_\infty$. Note that the critical $n$ for which
$\mathbf{X}\sim\mathcal{M}(n,\mathbf{p})$ belongs to the interior of $E_d$
corresponds to the classical ``coupon collector problem''
\cite{MR0228020,MR1344451,MR959649}. Another interesting open problem is the
optimality of the level-set regions related to the control of
$\mathbb{P}(\mathbf{p'}\in R_\alpha(\mathbf{X}))$ with
$\mathbf{X}\sim\mathcal{M}(n,\mathbf{p})$ and $\mathbf{p}\neq\mathbf{p'}$. It
might be also interesting to extend the level-set method to more complex
situations such as hierarchical log-linear models for instance.

\section*{Acknowledgements}

The present version of this article has greatly benefited from the comments
and criticism of an Associate Editor and three anonymous referees.

{
\bibliographystyle{plain} %
\addcontentsline{section}{References}\footnotesize %
\bibliography{biblio}
}

 \vfill
 {\noindent\footnotesize \textsc{Djalil \textsc{Chafa\"\i}, corresponding
     author, } \url{d.chafai[@]envt.fr}

   \medskip\noindent
   \textsc{UMR181 INRA, ENVT, \'Ecole Nationale V\'et\'erinaire de Toulouse \\
     23 Chemin des Capelles, F-31076 Cedex 3, Toulouse, France.}

   \medskip\noindent
   \textsc{UMR 5219 CNRS, Institut de Math\'ematiques,
     Universit\'e de Toulouse \\
     118 route de Narbonne, F-31062 Cedex 4, Toulouse, France.}}

 \vfill

\begin{figure}[p]
  \begin{center}
    \includegraphics[scale=0.3,angle=-90]{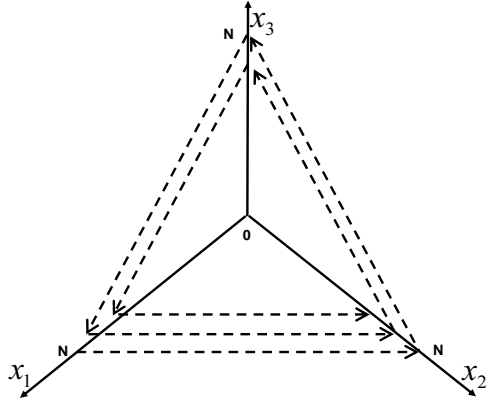}
    \caption{The construction of $A_k$ when $d=3$, with $A_0=\emptyset$ and
      $A_1=\{(n, 0, 0)\}$. The point in $A_1$ is at the beginning of the
      starting arrow represented as a dotted line. Each time the arrow meets a
      point in the simplex, this point is added to $A_k$ to give $A_{k+1}$.
      The set obtained with the three first arrows is invariant by permutation
      of coordinates.}
    \label{fig:crob0}
  \end{center}
\end{figure}

\begin{figure}[p]
  \begin{center}
    \includegraphics[scale=0.3,angle=0]{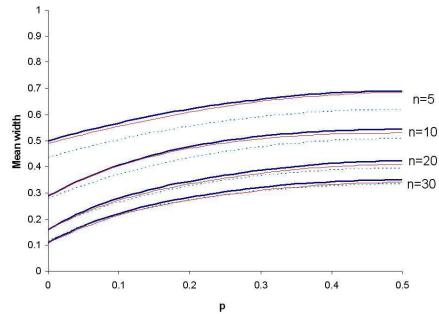}
    \caption{Binomial case $d=2$. The curves are the mean width of the
      $95\%$-intervals obtained with the Blyth-Still method (thick line), the
      level-set method (thin line) and the score method (dotted line) for
      $p_1\in[0, 0.5]$. The Blyth-Still method gives intervals with higher
      mean width irrespective of $p_1$. The score method always gives
      intervals with smaller width. Note that the score method fails to
      control the coverage probability. As $n$ increases, the differences
      between the mean widths of the respective intervals decrease.}
    \label{fig:crob1}
  \end{center}
\end{figure}

\begin{figure}[p]
  \begin{center}
    \includegraphics[scale=0.45,angle=0]{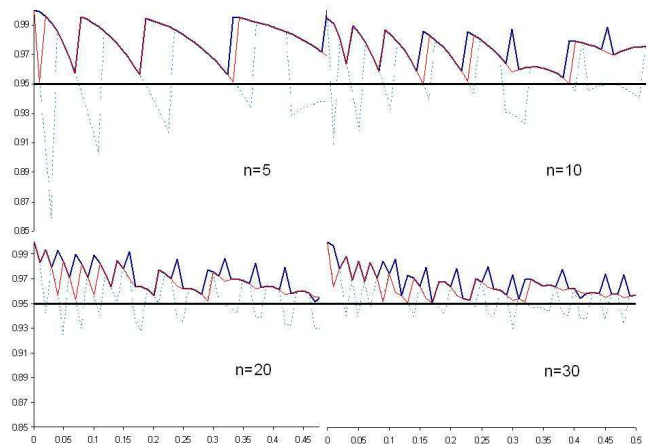}
    \caption{Binomial case $d=2$. These curves are the coverage of the
      $95\%$-intervals obtained with the Blyth-Still method (thick line), the
      level-set method (thin line) and the score method (dotted line) for
      $p_1\in[0, 0.5]$. The score method fails to control the coverage. The
      level-set method seems (nearly) uniformly better than the Blyth-Still
      method: its coverages are closer to $0.95$. When $n$ increases, the
      differences between these three methods decrease.}
    \label{fig:crob2}
  \end{center}
\end{figure}

\begin{figure}[p]
  \begin{center}
    \includegraphics[scale=0.28,angle=0]{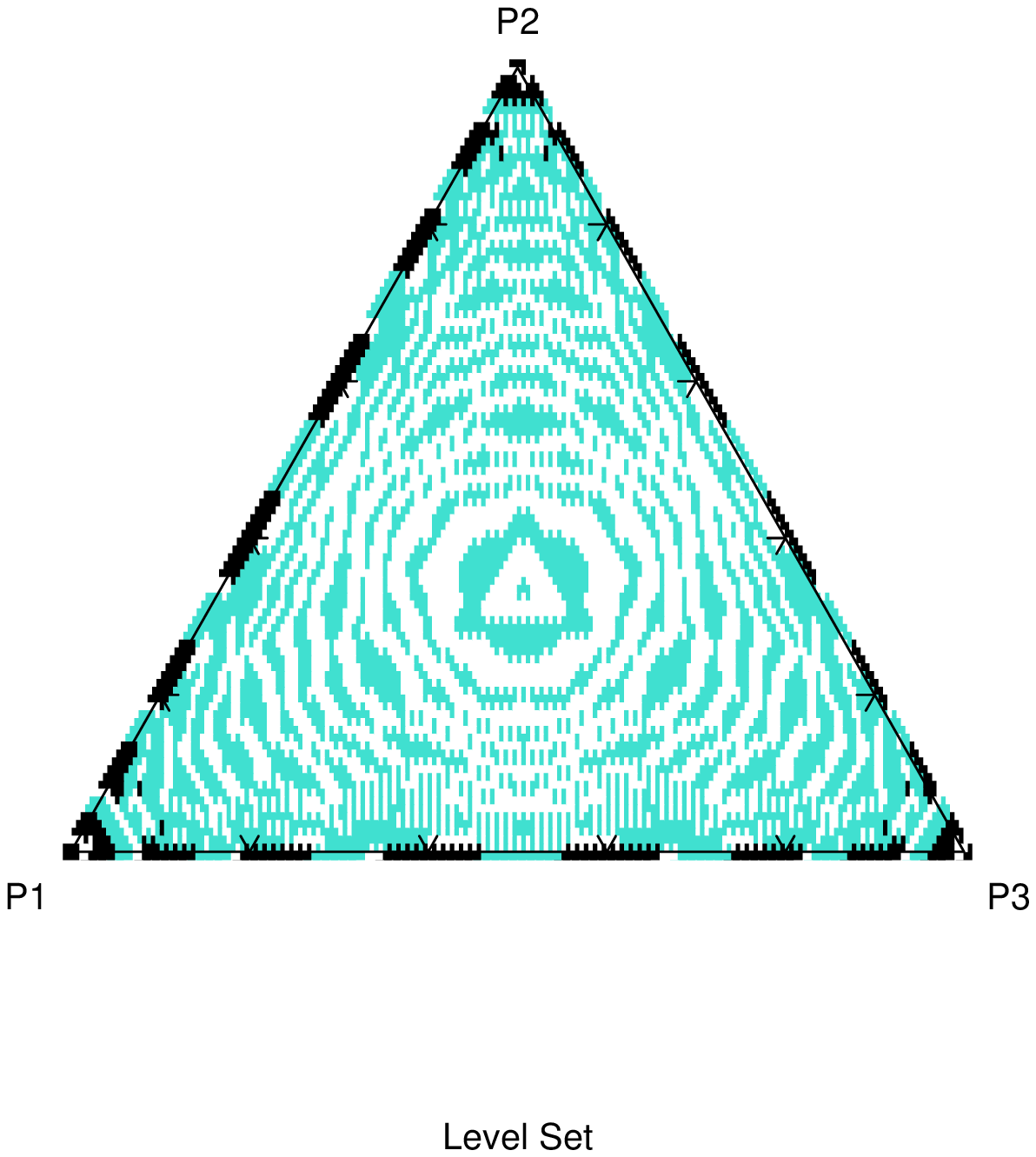}
    \includegraphics[scale=0.28,angle=0]{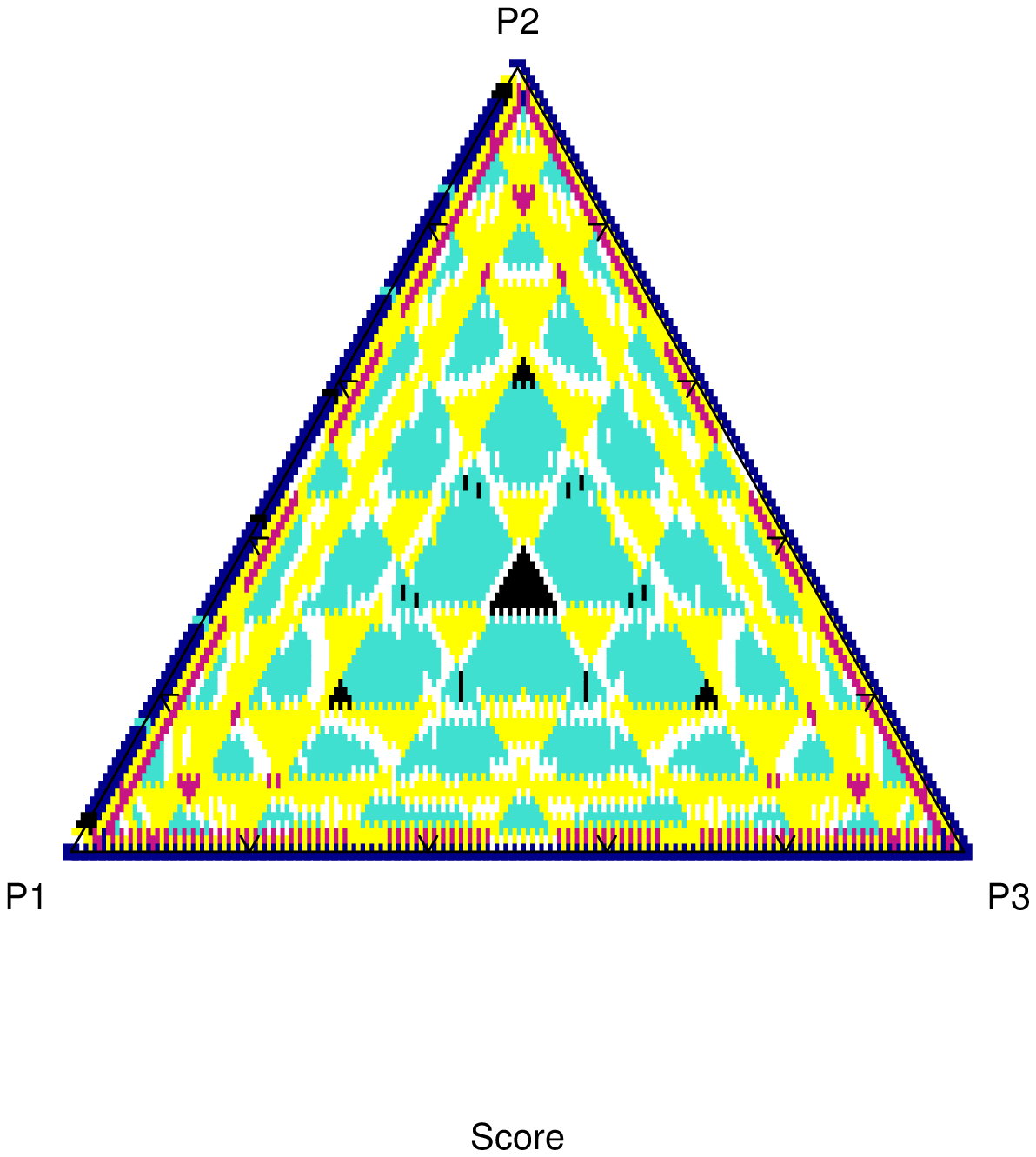}
    \includegraphics[scale=0.28,angle=0]{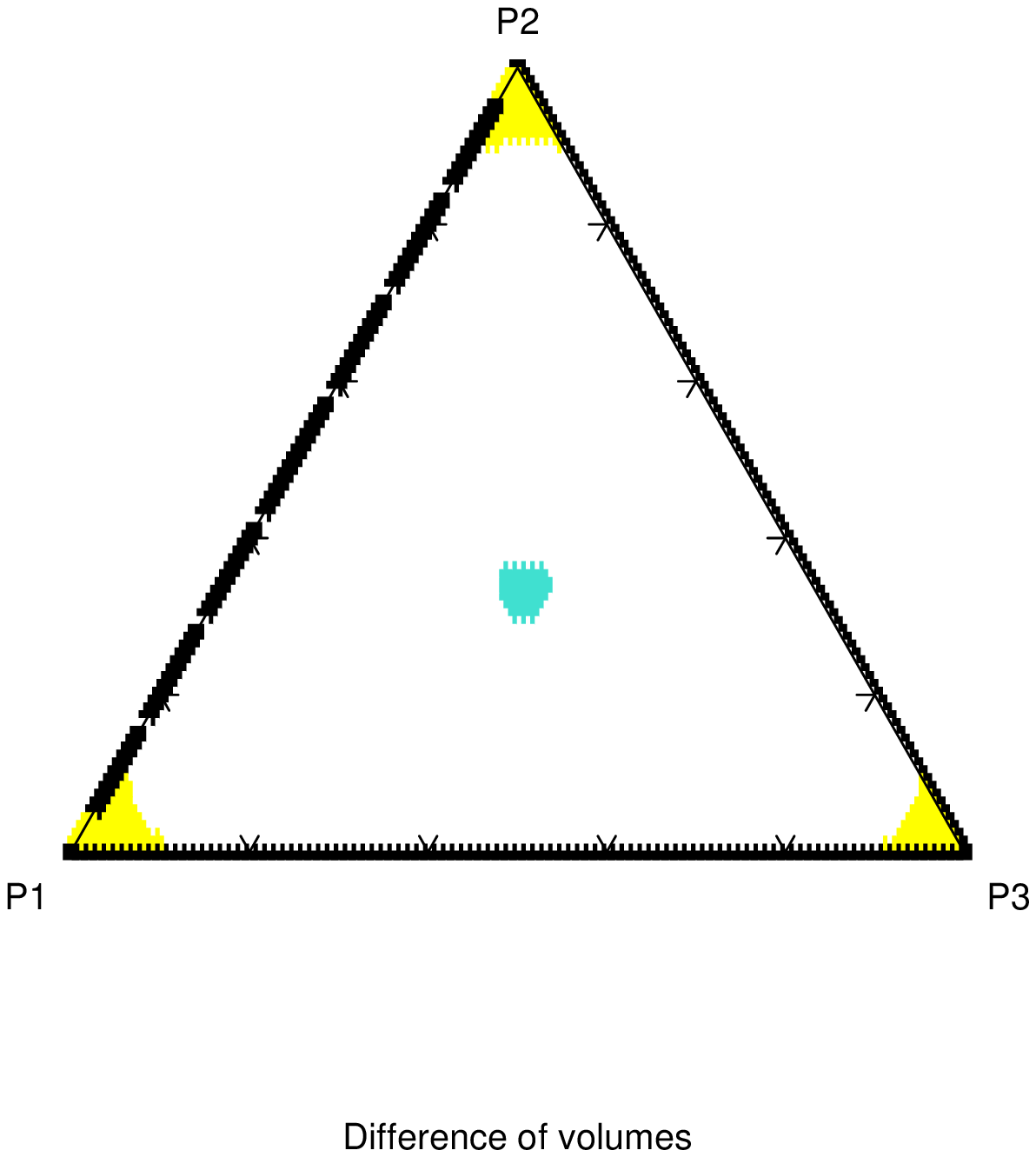}
  \end{center}
  \begin{center}
    \includegraphics[scale=0.28,angle=0]{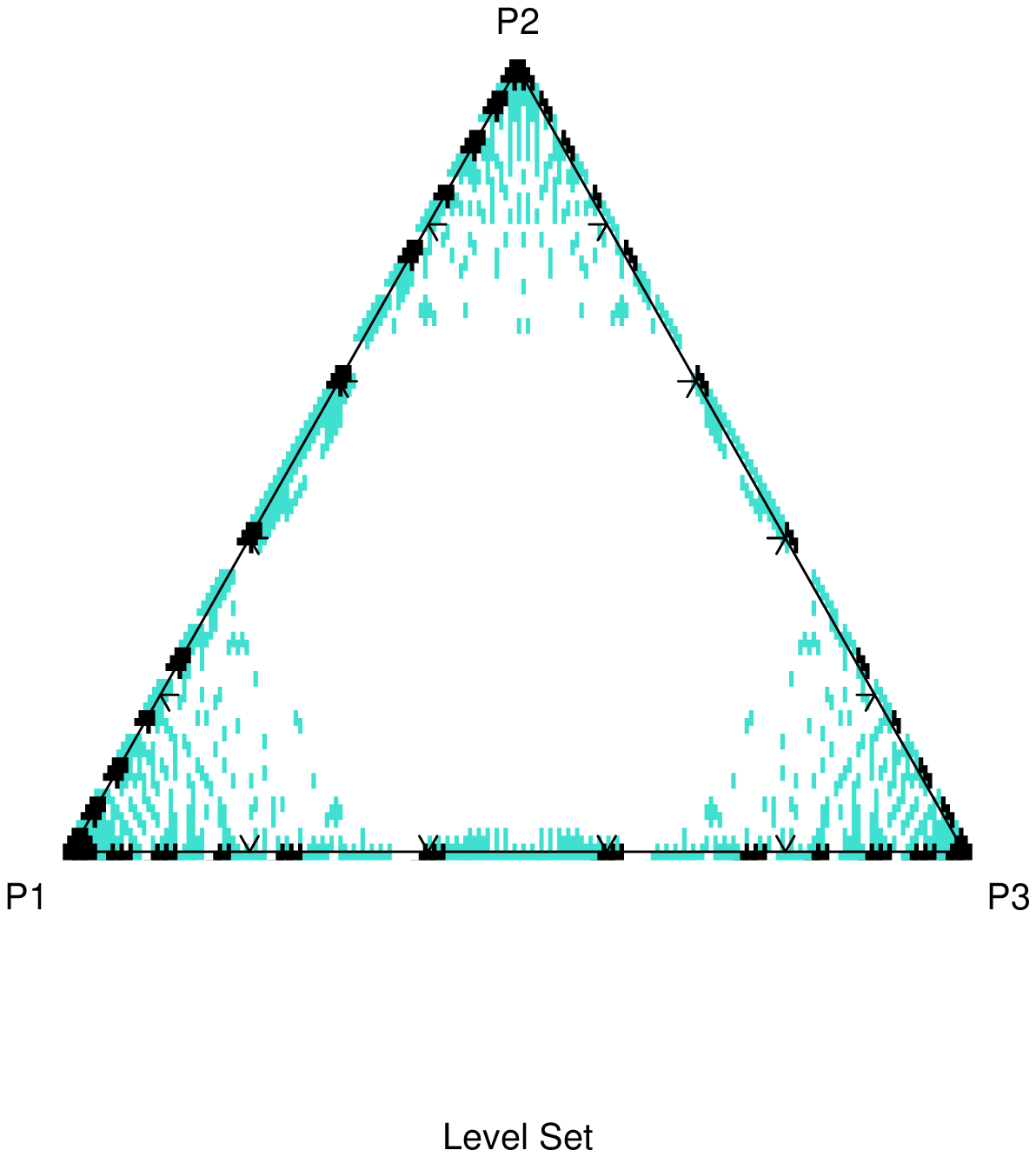}
    \includegraphics[scale=0.28,angle=0]{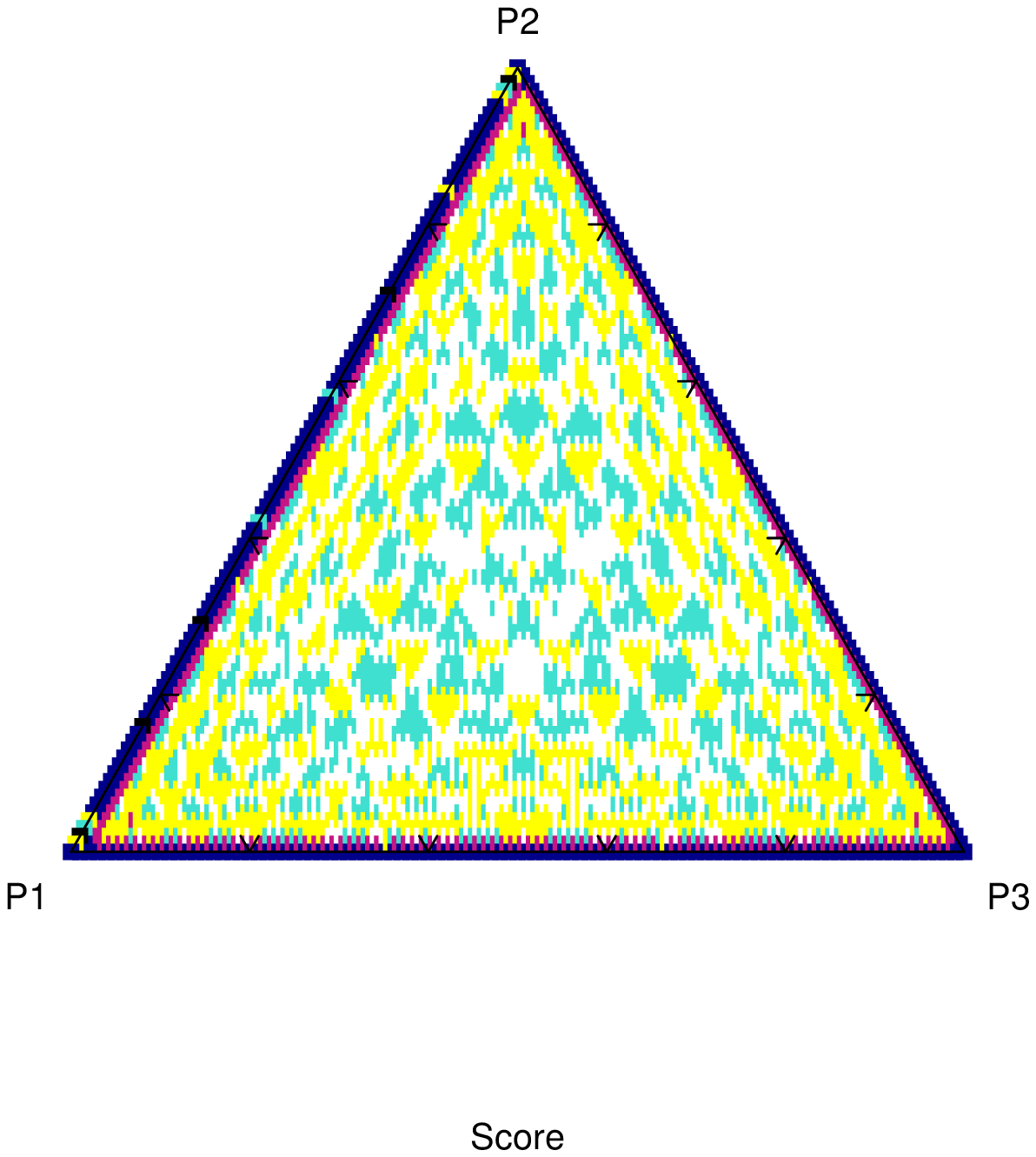}
    \includegraphics[scale=0.28,angle=0]{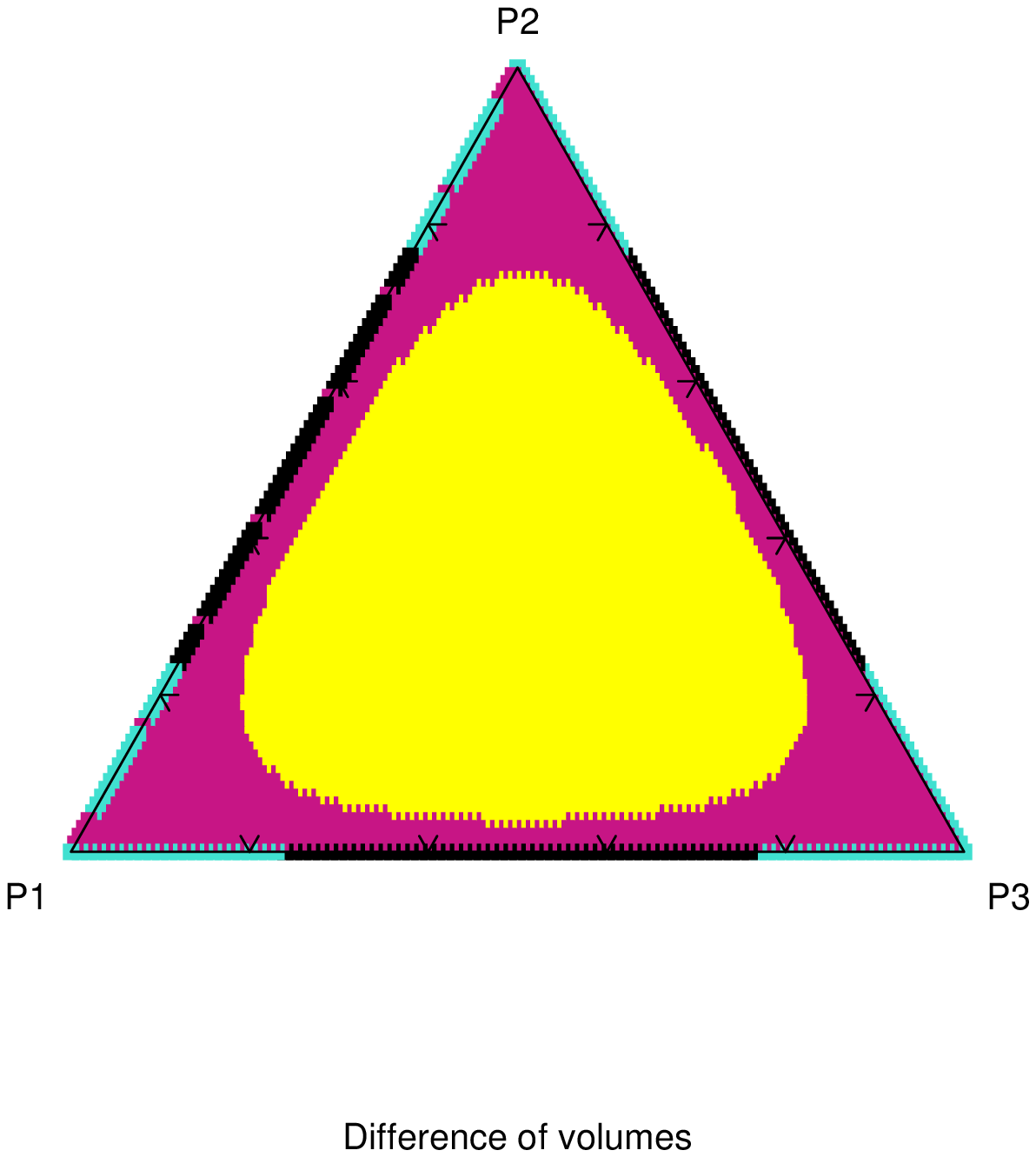}
  \end{center}
  \begin{center}
    \includegraphics[scale=0.28,angle=0]{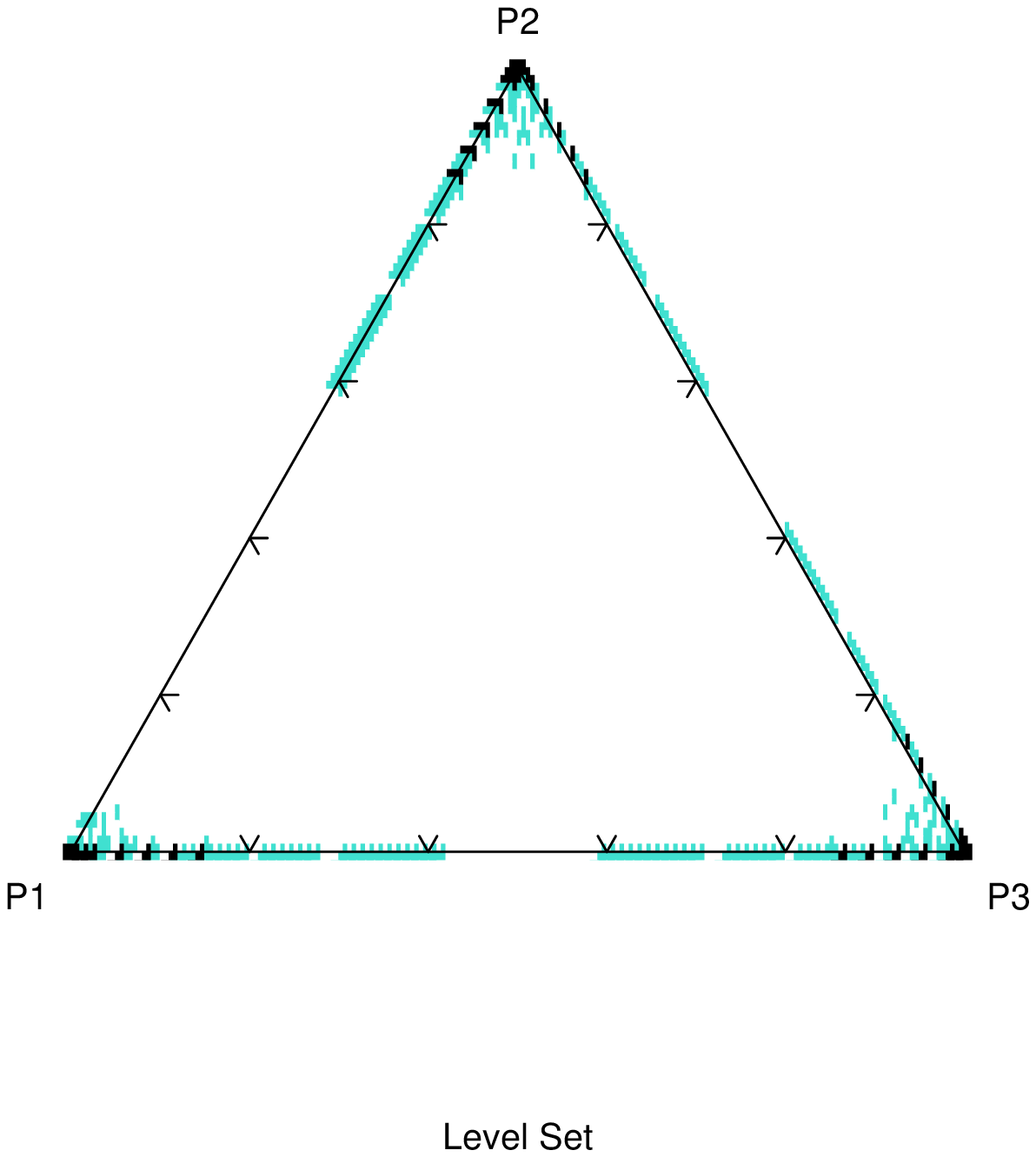}
    \includegraphics[scale=0.28,angle=0]{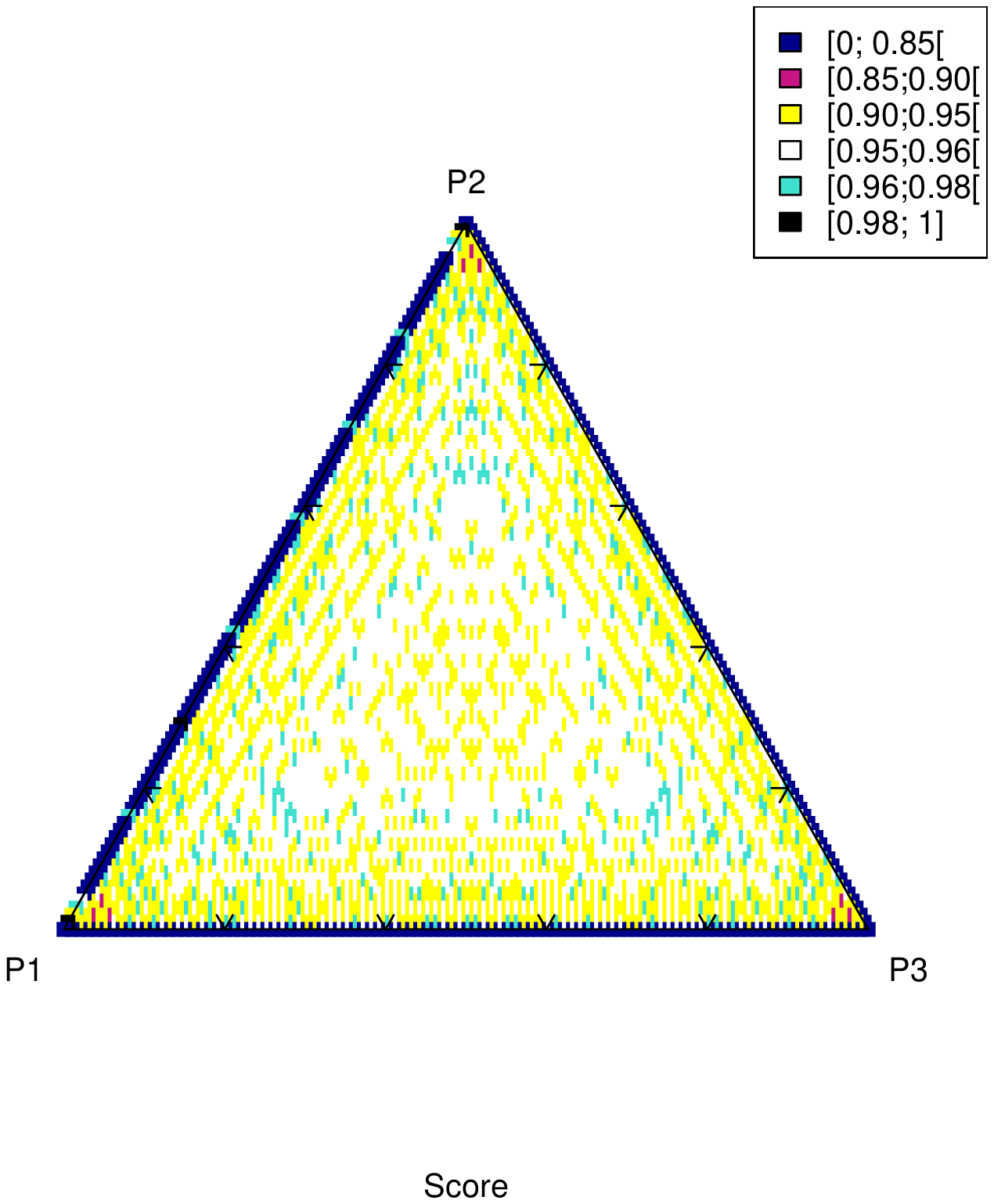}
    \includegraphics[scale=0.28,angle=0]{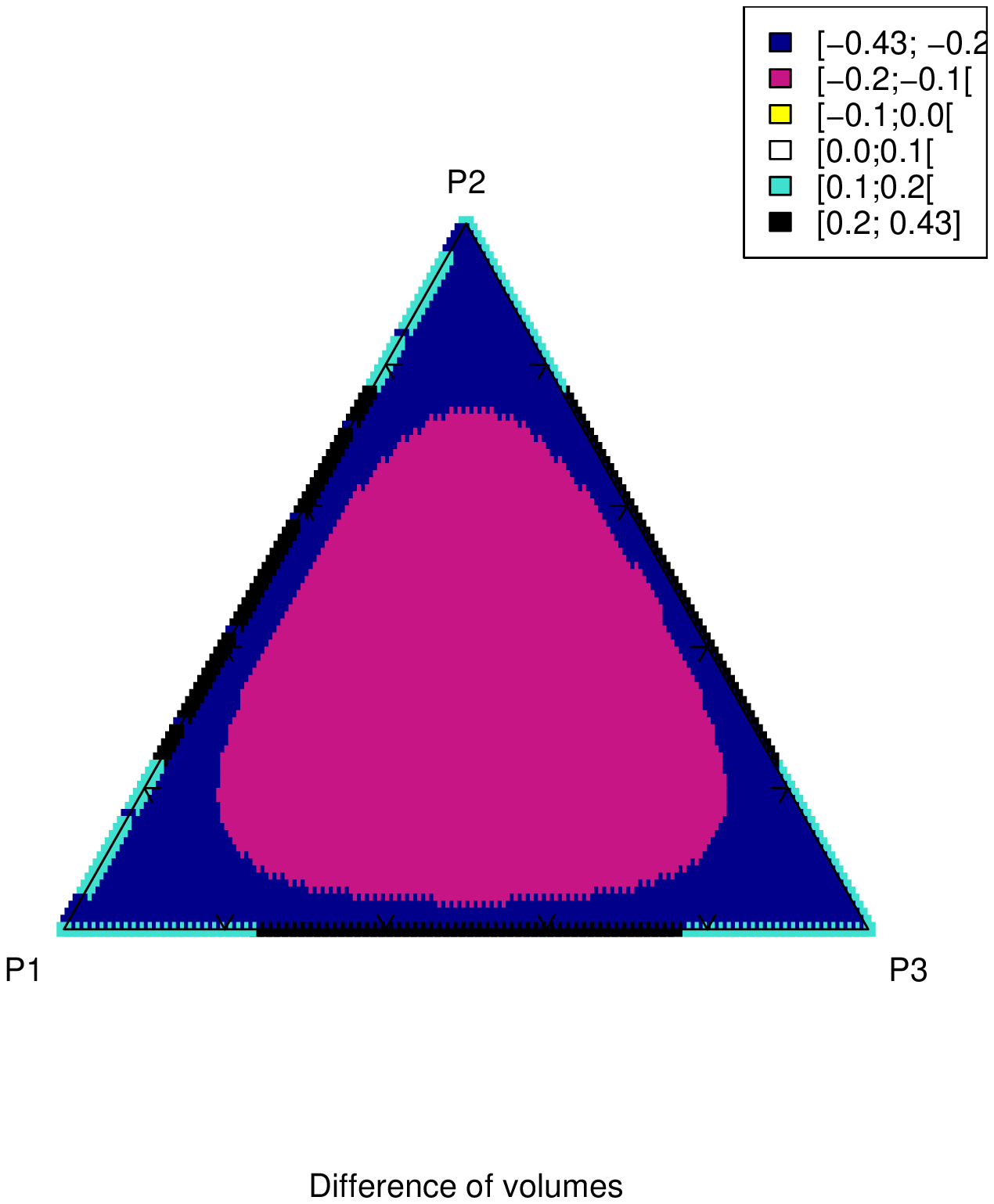}
    \caption{Trinomial case $d=3$. The columns give the coverages of the
      level-set method, the coverages given by the score method and the
      difference of mean volumes. The three rows correspond to
      $n\in\{5,10,20\}$. For the coverages graphs (first two columns), a clear
      color means that the coverage is close to $0.95$ whereas a dark blue
      color means that the coverage is smaller than $0.85$. For the volumes
      graphs (third column), a white color means that the difference of mean
      volumes is small whereas the blue, pink and yellow colors are used when
      the mean volume of the regions obtained with the level-set method are
      smaller than their counterpart obtained with the score method.}
       \label{fig:crob3}
  \end{center}
\end{figure}

\begin{figure}[p]
  \begin{center}
    \includegraphics[scale=0.45,angle=0]{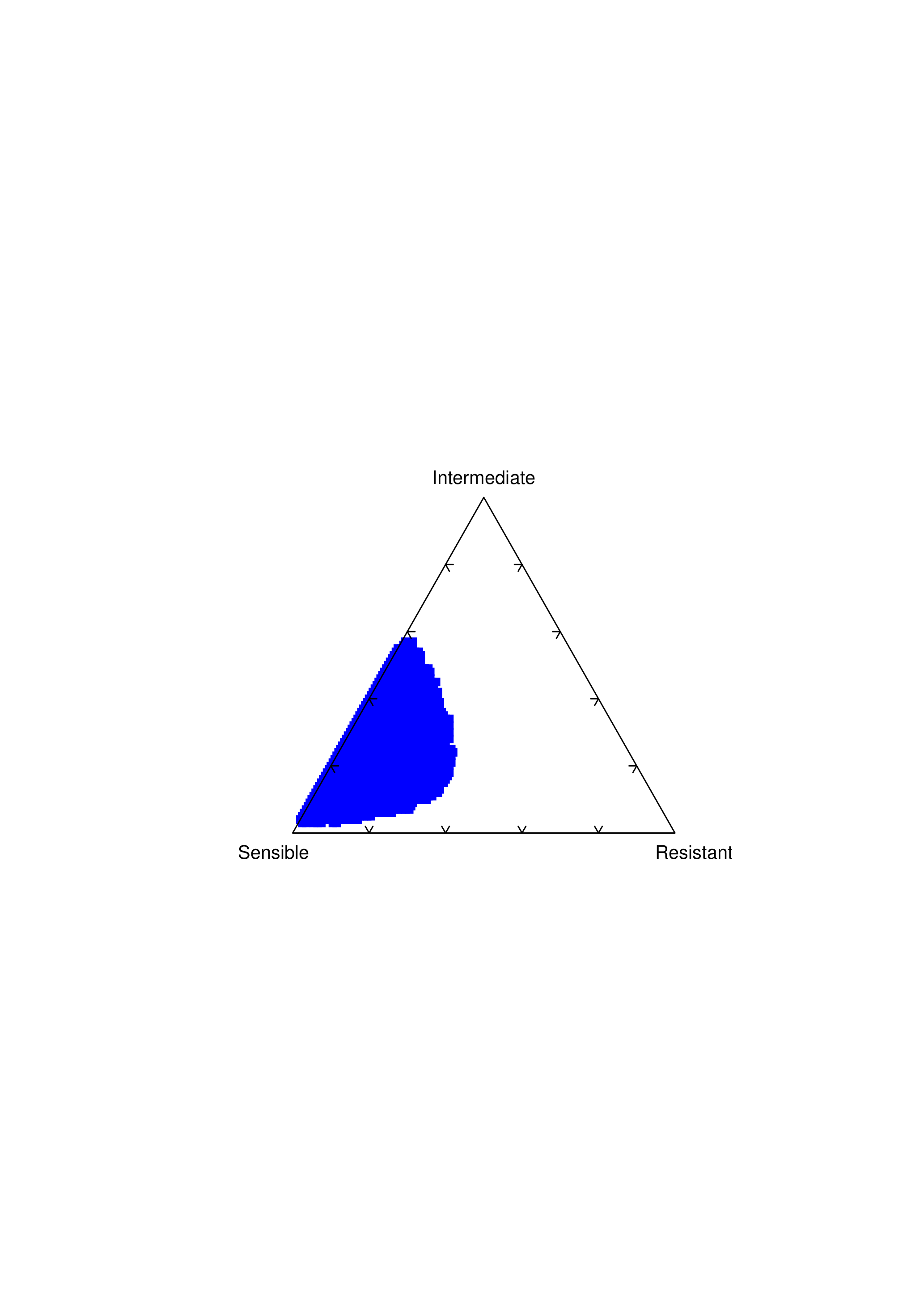}
    \caption{Trinomial case $d=3$ (example \ref{ss:ex1}). In barycentric
      coordinates, the $95\%$-region for $\mathbf{p}$ is constructed from the
      observation $\mathbf{x}=(0,2,8)$ of $\mathcal{M}_3(10,\mathbf{p})$. Note
      that the Wald method cannot be used here since the observation belongs
      to the boundary of the observation simplex $E_{3}$. In this example, the
      score and the level-set methods give approximately the same region.}
    \label{fi:ex1}
  \end{center}
\end{figure}

\begin{figure}[p]
  \begin{center}
    \includegraphics[scale=0.25,angle=0]{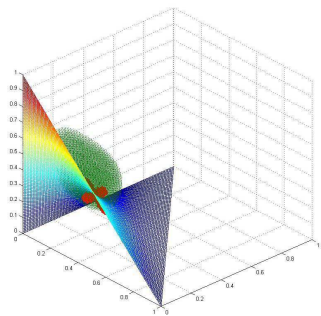}
    \caption{Quadrinomial case $d=4$ (example \ref{ss:ex2}). The axes
      correspond to $p_{1}$, $p_{2}$, and $p_{3}$. The null hypothesis $H_{0}$
      of the $\chi^{2}$-test is represented by the surface. The set in red is
      the acceptance region of the $\chi^{2}$-test. The region in green is the
      $95\%$-region for $\mathbf{p}$ built with the level-set method. It turns
      out that $\mathbf{\hat p}=(3/26, 8/26, 10/26)$does not belong to the
      acceptance region of the $\chi^{2}$-test while it belongs to the
      $95\%$-region for $\mathbf{p}$ built with the level-set method.
      Additionally, since this confidence region cuts $H_{0}$, the
      corresponding test does not reject $H_{0}$, in contrast to the
      $\chi^2$-test.}
    \label{fig:chi2}
  \end{center}
\end{figure}

\end{document}